\documentclass[letterpaper, 11pt]{article}

\usepackage{latexsym,amssymb,amsmath}
\usepackage{amsthm} 
\usepackage[left=1in,top=1in,right=1in,bottom=1in]{geometry}
\usepackage{hyperref}

\newcommand{\comment}[2]{}

\newcommand{\boxit}[1]{\begin{center}\framebox[0.85\textwidth]{\parbox{0.80\textwidth}{#1}}\end{center}}

\newtheorem{theorem}{Theorem}[section]
\newtheorem{lemma}[theorem]{Lemma}

\newtheorem{claim}[theorem]{Claim}
\newtheorem{observation}[theorem]{Observation}
\theoremstyle{remark}
\newtheorem{remark}[theorem]{Remark}

\theoremstyle{definition}

\newtheorem{reduction}[theorem]{Reduction}

\newcommand{\Real}{\mathbb R}
\DeclareMathOperator*{\E}{\mathbb{E}}

\newcommand{\eps}{\varepsilon}

\newcommand{\SD}{d_{\text{TV}}}

\newcommand{\G}{\mathcal{G}}

\title{Inapproximability of\\ NP-Complete Variants of Nash Equilibrium}
\author{Per Austrin \and Mark Braverman \and Eden Chlamt\'a\v{c}}

\date{University of Toronto, Toronto, Canada\\\texttt{\{austrin,mbraverm,eden\}@cs.toronto.edu}}

\begin{document}

\maketitle

\begin{abstract} In recent work of Hazan and Krauthgamer (SICOMP 2011), it was shown
that finding an $\eps$-approximate Nash equilibrium with near-optimal
value in a two-player game is as hard as finding a hidden clique of
size $O(\log n)$ in the random graph $G(n,\frac12)$.  This raises the
question of whether a similar intractability holds for approximate
Nash equilibrium without such constraints.  We give evidence that the
constraint of near-optimal value makes the problem distinctly harder:
a simple algorithm finds an optimal $\frac{1}{2}$-approximate equilibrium, while finding strictly better than $\frac12$-approximate equilibria
is as hard as the Hidden Clique problem.  This is in
contrast to the unconstrained problem where more sophisticated
algorithms, achieving better approximations, are known.



Unlike general Nash equilibrium, which is in PPAD, optimal (maximum
value) Nash equilibrium is NP-hard. We proceed to show that optimal
Nash equilibrium is just one of several known NP-hard problems related
to Nash equilibrium, all of which have approximate variants which are
as hard as finding a planted clique. In particular, we show this for
approximate variants of the following problems: finding a Nash
equilibrium with value greater than $\eta$ (for any $\eta>0$, even
when the best Nash equilibrium has value $1-\eta$), finding a second
Nash equilibrium, and finding a Nash equilibrium with small support.

Finally, we consider the complexity of approximate pure Bayes
Nash equilibria in two-player games.  Here we show that for general
Bayesian games the problem is NP-hard.  For the special case where the
distribution over types is uniform, we give a quasi-polynomial time
algorithm matched by a hardness result based on the Hidden Clique
problem.


\end{abstract}



\section{Introduction}
The classical notion of Nash equilibrium is the most fundamental
concept in the theory of non-cooperative games. In recent years, there
has been much work on the complexity of finding a Nash equilibrium in
a given game. In particular, a series of hardness results culminated
in the work of Chen et.\ al~\cite{CDT09}, who showed that even for
two-player (bimatrix) games, the problem of computing a Nash
equilibrium is PPAD-complete, and thus unlikely to be solvable in
polynomial time.

Therefore, it makes sense to consider the complexity of approximate
equilibria. In particular, a notion which has emerged as the focus of
several works is that of an \emph{$\eps$-approximate Nash
equilibrium}, or $\eps$-equilibrium for short, where neither player
can gain more than $\eps$ (additively) by defecting to a different
strategy (without loss of generality, all payoffs are scaled to lie in
the interval $[0,1]$).  A straightforward sampling argument of Lipton
et al.~\cite{LMM03} shows that in every game, there exist
$\eps$-equilibria with support $O(\log n/\eps^2)$, and so they can be
found in quasi-polynomial time $n^{O(\log n / \eps^2)}$ by exhaustive
search.

On the other hand, finding good polynomial time approximations has
proved more challenging. While finding a $\frac12$-equilibrium turns
out to be quite simple~\cite{DMP09}, more complicated algorithms have
given a series of improvements~\cite{DMP07,BBM10,TS08}, where the current best known is the
$0.3393$-equilibrium shown by Tsaknakis and Spirakis~\cite{TS08}.  A
major open question in this area is whether or not there exists a PTAS
for Nash Equilibrium (note that the algorithm of Lipton et al.\ gives
a quasi-polynomial time approximation scheme for the problem).

Recently, Hazan and Krauthgamer~\cite{HK11} have attempted to provide
evidence for the optimality of the QPTAS of Lipton et
al.~\cite{LMM03}, by showing a reduction from a well-studied and
seemingly intractable problem (which can also be solved in
quasi-polynomial time) to the related problem of finding an
$\eps$-equilibrium with near maximum value (the value of an
equilibrium is the average of the payoffs of the two players).


The problem they reduce from is the \emph{Hidden Clique Problem}:
Given a graph sampled from $G(n,\frac12)$ with a planted (but hidden)
clique of size $k$, find the planted clique.  Since with high
probability the maximum clique in $G(n,\frac12)$ is of size
$(2-o(1))\log n$, it is easy to see that for constant $\delta>0$, one
can distinguish between $G(n,\frac12)$ and $G(n,\frac12)$ with a
planted clique of size $k>(2+\delta)\log n$ in quasi-polynomial time
by exhaustive search over all subsets of $(2+\delta)\log n$
vertices. It is also not hard to extend this to an algorithm which
finds the hidden clique in quasi-polynomial time. 

On the other hand, the best known polynomial time algorithm, due to
Alon et al.~\cite{AKS98} only finds cliques of size $\Omega(\sqrt
n)$. In fact, Feige and Krauthgamer~\cite{FK03} show that even
extending this approach by using the Lov\'{a}sz-Schrijver SDP
hierarchy, one still requires $\Omega(\log n)$ levels of the hierarchy
(corresponding to $n^{\Omega(\log n)}$ running time to solve the SDP)
just to find a hidden clique of size $n^{1/2-\eps}$. The only possible
approach we are aware of for breaking the $\Omega(\sqrt{n})$ barrier
would still (assuming certain conjectures) only discover cliques of
size $\Omega(n^c)$ for some constant $c>0$~\cite{FK08,BV09}.

Hazan and Krauthgamer show that finding a near-optimal
$\eps$-equilibrium is as hard as finding hidden cliques of size
$C \log n$, for some universal constant $C$.  Here, by near-optimal we
mean having value close to maximum possible value obtained in an
actual Nash equilibrium.  Subsequently, Minder and
Vilenchik~\cite{MV09} then improved this hardness to planted cliques
of size $(2+\delta) \log n$ for arbitrarily small $\delta >
0$.\footnote{There is a small caveat: the reduction of~\cite{MV09}
only certifies the presence of a hidden clique (i.e.\ distinguishes
the graph from the random graph $G(n,\frac12)$ w.h.p.)  but does not
identify the vertices of the clique.} Here, we will rely on the
hardness assumption for hidden cliques of size $C\log n$ for any
constant $C$, and will not concern ourselves with optimizing the value
of $C$.

\subsection{A Sharp Result For Near-Optimal Approximate Nash}

It is important to note that the problem considered in~\cite{HK11} is
not equivalent to finding an unconstrained $\eps$-equilibrium.  In
light of the results of \cite{HK11,MV09} it is natural to ask to what
extent the hardness for near-optimal approximate equlibrium gives an
indication of hardness for unconstrained approximate equilibrium.
Indeed, \cite{HK11}, in their concluding remarks, ask whether their
methods can be used to rule out a PTAS for unconstrained Nash
equilibrium.  One of the messages of this paper is that these two
problems are quite different in terms of approximability and that one
should not yet be overly pessimistic about the possibility for a PTAS
for unconstrained Nash equilibrium.  Indeed, while there is a a
polynomial time algorithm to find a $0.3393$-equilibrium, we show that
finding a near-optimal $(\frac{1}{2}-\eta)$-equilibrium is hard.

\begin{theorem}[Informal]
  \label{thm:eps-hardness} 
  For every $\eta>0$, finding a near-optimal
  $(\frac{1}{2}-\eta)$-approximate equilibrium is as hard as finding a hidden
  clique of size $C \log n$ in $G(n, \frac{1}{2})$.
\end{theorem}

As mentioned above, there is a simple polynomial time algorithm to
find a $\frac12$-equilibrium, and we show that this algorithm can be
extended to find a $\frac{1}{2}$-equilibrium with value at least that
of the best exact equilibrium:

\begin{theorem}[Informal]
  \label{thm:1/2-algo}
  There exists a polynomial time algorithm to find a
  $\frac{1}{2}$-approximate equilibrium with value at least that of
  the optimal true equilibrium.
\end{theorem}

Thus, Theorem~\ref{thm:eps-hardness} is tight and unlike unconstrained
Nash equilibrium, where stronger techniques yield approximations better than $\frac{1}{2}$, near-optimal Nash equilibrium does not admit
efficient ``non-trivial'' approximations (assuming the Hidden Clique
problem is hard).

\subsection{The Bigger Picture: Hardness for NP-hard Variants of Nash}

Just like with unconstrained $\eps$-equilibrium, finding a
near-optimal $\eps$-equilibrium can be done in quasi-polynomial time
using the algorithm of~\cite{LMM03}.  However, the exact version --
finding a maximum value Nash equilibrium -- is known to be
NP-hard~\cite{GZ89} and therefore harder than its unconstrained
counterpart which is in PPAD~\cite{Papadimitriou94}.  In fact, maximum value Nash is one of
several optimization variants of Nash equilibrium which are
NP-complete.  Other variants include: determining whether a bimatrix
game has more than one Nash equilibrium~\cite{GZ89}, finding a Nash
Equilibrium with minimum support~\cite{GZ89}, and determining whether
there exists an equilibrium with value at least $1-\frac1n$ or all
equilibria have value at most $\eps/n$ (even for arbitrarily small
$\eps=\eps(n)>0$)~\cite{CS08}.  We show that approximate-equilibrium
variants of these problems are also as hard as Hidden Clique.  

For the problem of obtaining any non-trivial approximation to the
optimal value of a Nash equilibrium, we prove the following theorem.

\begin{theorem}[Informal]
  \label{thm:delta-hardness} For every $\eta>0$, finding an
  $\eps$-equilibrium with value at least $\eta$ is as hard as finding
  a hidden clique of size $C \log n$ in $G(n, \frac{1}{2})$, even in a game
  having an equilibrium of value $1-\eta$.
\end{theorem}

For the case of determining whether a game has more than one
equilibrium, note that by continuity considerations, every two-player
game has an infinite number of $\eps$-equilibria.  Thus, the
appropriate approximate analog is to consider the problem of finding
two $\eps$-equilibria with (at least) a certain total variation
distance between them.  We show that this problem is also as hard as
Hidden Clique.

\begin{theorem}[Informal]
  \label{thm:second_equi}
  For all sufficiently small $\epsilon > 0$, determining whether a
  game has two different $\eps$-approximate equilibria (say, having
  statistical distance at least $3\epsilon$) is as hard as finding a
  hidden clique of size $C \log n$ in $G(n, \frac{1}{2})$.
\end{theorem}

We then move to the problem of finding an equilibrium with small
support.  Recall that by~\cite{LMM03}, there exist $\eps$-Nash
equilibria with support $O(\log n/\eps^2)$. It is also known that for
any $\eta>0$, in certain two-player games all
$(\frac12-\eta)$-equilibria must have support at least $\log
n/(1+\log(1/\eta))$~\cite{FNS07} (the threshold of $\frac12$ is tight,
since the simple $\frac12$-equilibrium of~\cite{DMP09} has support
3). As an approximate-equilibrium variant of the Minimum Support
Equilibrium problem, we consider the problem of finding an
$\eps$-equilibrium with support at most some threshold $t$, and prove
the following hardness result.


\begin{theorem}[Informal]
  \label{thm:small_support} 
  For every $\eta > 0$, finding a
  $(\frac{1}{2}-\eta)$-equilibrium with support size $C' \log n$ is as hard as
  finding a hidden clique of size $C \log n$ in $G(n, \frac{1}{2})$.
\end{theorem}

This can be seen as a complexity-theoretic analogue of the lower bound of~\cite{FNS07} mentioned above. Again, this contrasts with the situation for unconstrained Nash equilibrium, which is guaranteed to exist, and admits stronger approximations.

While these are all negative results, we again would like to stress
that there is a positive message to this story: these problems are
hard because they are approximate versions of NP-complete problems,
not because they are approximate variants of Nash equilibrium.
Therefore, these results should \emph{not} be viewed as indications
that Nash equilibrium does not have a PTAS.

\subsection{The Complexity of Approximate Pure Bayes Nash Equilibria}

Finally, we consider the problem of approximating pure \emph{Bayes
Nash Equilibria} (BNE) in two-player games.  \emph{Bayesian games} model the
situation where the players' knowledge of the world is incomplete.  
In a Bayesian game, both players may be in one of a number of
different states, known as \emph{types}, representing what each player
knows about the state of the world, and the payoff of each player
depends on the type of both players in addition to their strategies.
The types are distributed according to some joint distribution and are
not necessarily independent.  A pure strategy for a Bayesian game
assigns to each type a strategy that the player plays when she is in
that type.  In a pure BNE, conditioning on
any given type for a given player, the player cannot gain by changing
his strategy for that type.  See Section~\ref{sec:bayes} for precise
definitions.

Conitzer and Sandholm~\cite{CS08} have shown that determining whether
a given two-player game has a pure BNE is NP-complete.  We show that
this holds also for approximate pure BNE.

\begin{theorem}[Informal]
\label{thm:BNE-NP-hard}
Let $\eps=0.004$.  Then given a Bayesian game that admits a pure BNE,
it is NP-hard to find a pure $\eps$-BNE for the game.
\end{theorem}

However, this hardness result relies heavily on the joint distribution
of the players' types being non-uniform (in fact, not even product distribution).  We show that
when the distribution over type pairs is uniform, there is in fact a
quasi-polynomial time algorithm for $\eps$-approximate pure BNE (when
a pure BNE exists).

\begin{theorem}[Informal]
\label{thm:BNE-algo} 
For every $\eps > 0$ there is a quasipolynomial time algorithm to find
a pure $\eps$-BNE in two-player Bayesian games with uniformly
distributed types and in which a pure BNE exists.
\end{theorem}

We remark that this algorithm extends easily to arbitrary product
distributions over types but in order to keep the notation simple we
restrict our attention to the uniform case.

The algorithm is tight: it follows immediately from our hardness for
Small Support Equilibrium that this problem is also as hard as Hidden
Clique. 

\begin{theorem}[Informal]
  \label{thm:BNE-PC-hard} For every $\eta>0$, finding a $(\frac14-\eta)$-approximate pure BNE in a two-player Bayesian games with uniformly
distributed types and in which a pure BNE exists  
 is as hard as finding a hidden clique of size $C\log n$.
\end{theorem}

\subsection{Organization}

In Section~\ref{sec:goodvalue} we prove the results relating to
approximate equilibria with good value: Theorem~\ref{thm:eps-hardness}
(Section~\ref{sec:eps-hardness}), Theorem~\ref{thm:1/2-algo}
(Section~\ref{sec:1/2-algorithm}), and
Theorem~\ref{thm:delta-hardness} (Section~\ref{sec:delta-hardness}).

In Section~\ref{sec:uniqueness} we prove Theorem~\ref{thm:second_equi}
by a black-box application of Theorem~\ref{thm:delta-hardness}.  In
Section~\ref{sec:small-support} we prove
Theorem~\ref{thm:small_support} using similar techniques as for the
hardness results of Section~\ref{sec:goodvalue}.  In
Section~\ref{sec:bayes} we prove our results for Bayesian games,
Theorems~\ref{thm:BNE-NP-hard},~\ref{thm:BNE-algo},~and~\ref{thm:BNE-PC-hard}.

\section{Preliminaries}


A \emph{bimatrix game} $\G = (M_{\mathrm{row}}, M_{\mathrm{col}})$ is
a game defined by two finite matrices, $M_{\mathrm{row}}$
and $M_{\mathrm{col}}$, and two players: the \emph{row player} and the
\emph{column player}. We assume throughout that the game is normalized
so that both matrices have values in the interval $[0,1]$. The row and
column players choose \emph{strategies} $x$ and $y$ respectively,
where $x,y$ are nonnegative vectors satisfying
$\sum_ix_i=\sum_jy_j=1$.  A \emph{pure} strategy is one with support
$1$ (i.e.\ a vector with $1$ in one entry and $0$ in the rest). The
row (resp.\ column) player's \emph{payoff} is given by
$x^{\top}M_{\mathrm{row}}y$ (resp.\ $x^{\top}M_{\mathrm{col}}y$).

A \emph{Nash equilibrium} is a pair of strategies $(x,y)$ such that
neither player has any incentive to deviate to a different strategy,
assuming the other player does not deviate. Formally, in an
equilibrium, for all $i,j$ we have $e_i^{\top}M_{\mathrm{row}}y\leq
x^{\top}M_{\mathrm{row}}y$ and $x^{\top}M_{\mathrm{col}}e_j\leq
x^{\top}M_{\mathrm{col}}y$. An $\eps$-approximate Nash equilibrium, or
\emph{$\eps$-equilibrium} for short, is a pair of strategies $x,y$
where each player has incentive at most $\eps$ to deviate. That is,
for all $i,j$, $$e_i^{\top}M_{\mathrm{row}}y\leq
x^{\top}M_{\mathrm{row}}y+\eps\qquad\text{and}\qquad
x^{\top}M_{\mathrm{col}}e_j\leq x^{\top}M_{\mathrm{col}}y+\eps.$$

The \emph{value} of a pair of strategies, denoted $v_{\G}(x,y)$, is
the average payoff of the two players, i.e., 
$$v_\G(x,y) = \frac{1}{2}(x^{\top}M_{\mathrm{row}}y + x^{\top}M_{\mathrm{col}}y) = \sum_{i,j} x_i y_j \frac{M_{\mathrm{row}}(i,j) + M_{\mathrm{col}}(i,j)}{2}.$$

For a vector $x \in \Real^n$ and $S \subseteq [n]$, we write $x_S$ for
the projection of $x$ to the coordinates $S$.  We write $\|x\| =
\sum_{i = 1}^n |x_i|$ for the $\ell_1$ norm of $x$.  Thus, for a
strategy (in other words, a probability distribution) $x \in [0,1]^n$
we write $\|x_S\|$ for the probability that the player plays an
element of $S$.

Further, for a set $S \subseteq [n]$ of strategies, we use
$v_{\G|S}(x,y)$ to denote the value of $(x,y)$ \emph{conditioned} on
both players playing in $S$.  Formally,
$$ v_{\G|S}(x,y) = \E_{i \sim x, j \sim y}\left[
  \frac{M_{\mathrm{row}}(i,j)+M_{\mathrm{col}}(i,j)}{2}\,\Big|i,j \in S\right] = \frac{v_\G(x_S, y_S)}{\|x_S\| \cdot \|y_S\|}.$$
(If $\|x_S\| = 0$ or $\|y_S\| = 0$, $v_{\G|S}(x,y)$ is undefined.)

Given an undirected graph $G=(V,E)$, and (not necessarily disjoint) vertex sets $S_1,S_2\subseteq V$, we will denote by $E(S_1,S_2)$ the set of ordered pairs $\{(i,j)\in S_1\times S_2\mid \{i,j\}\in E\text{ or }i=j\}$. We will refer to $d(S_1,S_2)=|E(S_1,S_2)|/(|S_1||S_2|)$ as the \emph{density} of the pair $(S_1,S_2)$.

Finally, we shall make repeated use of the following standard Chernoff
bound.

\begin{lemma}[Chernoff bound]
  \label{lemma:chernoff}
  Let $X_1, \ldots, X_m$ be i.i.d.\ $\{0,1\}$ random variables and let
  $\mu = \frac{1}{m} \E[ \sum X_i]$.  Then
  $$
  \Pr\left[ \frac{1}{m} \sum X_i \ge \mu + \eps \right] \le e^{-2\eps^2 m}.
  $$
\end{lemma}

\section{Approximate Equilibria With Good Value}
\label{sec:goodvalue}

\subsection{The Reduction}\label{sec:reduction}

In this section we describe the general reduction that we use to prove
Theorems~\ref{thm:eps-hardness} and \ref{thm:delta-hardness} and
describe its properties. This reduction also forms the basis for the reductions we use to prove Theorems~\ref{thm:second_equi}, \ref{thm:small_support}  and \ref{thm:BNE-PC-hard}. It is based on the reduction of \cite{HK11}.

As in \cite{HK11} our soundness analysis proceeds by using the approximate equilibrium to find a dense
bipartite subgraph of $G$. The following lemma shows 
that this is sufficient to recover the hidden clique.

\begin{lemma}[{\cite[Lemma 5.3]{HK11}}]
  \label{lemma:reconstruct}
  There exist universal constants $c_1$ and $c_2$ such that the
  following holds.  Let $G$ be a sample from $G(n, \frac{1}{2})$ with a
  hidden clique of size $C \log n$ for some $C \ge c_1$.
  Then, given a pair of vertex sets $S_1,S_2 \subseteq [n]$ of size $c_2
  \log n$ and density $d(S_1,S_2)\geq5/9$ we can in polynomial time reconstruct the
  hidden clique (with high probability over $G$).
\end{lemma}

The lemma is slightly different from Lemma 5.3 of \cite{HK11}: there
we start with a bipartite subgraph of density $3/5$ instead of $5/9$
but this minor difference only changes the value of the constant $c_2$
-- the lemma holds for any constant density strictly larger than
$\frac{1}{2}$.

Let us now describe the reduction.  It is controlled by three
parameters $\alpha, \beta, \gamma \in (0,1)$.  Setting these parameters
  appropriately gives the various hardness results.

\boxit{
\begin{reduction}
  \label{red:reduction}
  Let $G = (V,E)$ be an $n$ vertex graph and $A$ its adjacency matrix
  (with $1$s on the diagonal).  Then, for parameters $\alpha, \beta, \gamma \in (0,1)$, we define a (random) bimatrix game $\G := \G(G,
    \alpha, \beta, \gamma)$ as follows.

  Let $N = n^c$ where $c = (c_2+1) \log 1/\beta$ for the universal
  constant $c_2$ of Lemma~\ref{lemma:reconstruct}.  Pick a random $N
  \times n$ matrix $B$ whose entries are i.i.d.\ $\{0,1\}$ variables
  with expectation $\beta$.  Then $\G = (M_{\mathrm{row}},
  M_{\mathrm{col}})$, where the payoff matrices are:
  \begin{align}\label{eq:simple-payoffs}
    M_{\mathrm{row}}=\left(
    \begin{array}{cc}
      \alpha A & 0\\
      B &\gamma J
    \end{array} 
    \right)
    & &
    M_{\mathrm{col}}=\left(
    \begin{array}{cc}
      \alpha A & B^{\top}\\
      0 &\gamma J
    \end{array} 
    \right),
  \end{align}
  where $J$ is the all-ones $N \times N$ matrix.
\end{reduction}
}


We conclude this section with an additional lemma 
which shows how to obtain a dense bipartite subgraph given an
approximate equilibrium of $\G$ with certain properties.  This lemma
(and its proof) is analogous to~\cite[Lemma 5.2]{HK11}, but as we need
it in larger generality we also give the proof.

\begin{lemma}
  \label{lem:dense-subgraph}
  Let $\G$ be as in Reduction~\ref{red:reduction}.  Fix any $s \in
  [0,1]$, $t \in [0,1]$ and $\eps \in [0,1]$ such that $1 - t - 3
  \sqrt{s}/2 \ge \alpha + \eps$, and let $(x, y)$ be an
  $\eps$-approximate equilibrium of $\G$ with the following two
  properties:
  \begin{itemize}
  \item Both $\|x_{[n]}\| \ge 1-t$ and $\|y_{[n]}\|
    \ge 1-t$.
  \item The conditional value $v_{\G|[n]}(x,y) \ge (1-s)\alpha$.
  \end{itemize}
  Then, given $(x,y)$ as above, we can efficiently find vertex sets
  $S_1, S_2 \subseteq [n]$ each of size $c_2 \log n$ and density $d(S_1,S_2)\ge 5/9$.
\end{lemma}

In the proof of Lemma~\ref{lem:dense-subgraph} we shall use the
following simple claim, which is analogous to~\cite[Claim 5.3]{HK11}:

\begin{claim}
  Let $S\subset[n]$ be a set of size $|S| \le c_2 \log n$. Then,
  w.h.p.\ over $\G$, in any $\eps$-equilibrium in the above game, the
  probability mass a column (or row) player may place on such a set
  $S$ is at most $\eps+\alpha$.
\end{claim}

The proof of the claim follows by noting that for such $S$, with high
probability, there is a row $i$ of $B$ such that $B_{ij} = 1$ for all
$j \in S$.  Indeed, the probability that such a row does not exist is
exactly
$$
\left(1 - \beta^{-|S|}\right)^N \le \exp(-n^{-c_2 \log 1/\beta} N) =
\exp(-n^{\log 1/\beta}).
$$ 
We omit the remaining details as they are identical to~\cite[Claim
  5.3]{HK11}.


\begin{proof}[Proof of Lemma~\ref{lem:dense-subgraph}]
  Let $\tilde{x}, \tilde{y} \in [0,1]^n$ be the strategies $x$ and $y$
  conditioned on playing in $[n]$.  That is,
  $\tilde{x}=x_{[n]}/\|x_{[n]}\|$, and
  $\tilde{y}=y_{[n]}/\|y_{[n]}\|$.  The second given property of the
  strategy pair $(x,y)$ can be rephrased as $\tilde{x}^{\top} A
  \tilde{y} \ge 1-s$.
  
  Let
  $$S'_1 = \left\{ i \in [n]: e_i^{\top} A \tilde{y} \ge
  1-2\sqrt{s}/3\right\}.$$ To obtain a lower bound on the cardinality $|S'_1|$, we shall bound
  $\|x_{S'_1}\|$ by $\eps + \alpha$ and then appeal to the claim.  By the first given property of $(x,y)$, we have 
  $$\|x_{S'_1}\| \ge 1-t - \|x_{[n] \setminus S'_1}\| \ge 1 - t - \|\tilde{x}_{[n]\setminus  S'_1}\|.$$
  We can bound the last term, $\|\tilde{x}_{[n] \setminus S'_1}\|$, from above using Markov's inequality, viz.,
  $$
  \|\tilde{x}_{[n] \setminus S'_1}\| = \Pr_{i \sim \tilde{x}}\left[1 - e_i^{\top} A \tilde{y} > 2\sqrt{s} /3\right] \le \frac{s}{2 \sqrt{s} / 3} = 3\sqrt{s}/2.
  $$ Thus we have $\|x_{S'_1}\| \ge 1 - t - 3\sqrt{s}/2 \ge \eps +
  \alpha$ and so by the claim, $|S'_1| \ge c_2 \log n$.  Truncate $S'_1$
  by taking any arbitrary subset $S_1\subseteq S'_1$ of cardinality
  $|S_1|=c_2 \log n$.
  
  Now let $\overline{x} \in [0,1]^n$ be the uniform distribution over
  $S_1$.  Note that $\overline{x}^{\top} A \tilde{y} \ge
  1-2\sqrt{s}/3$.  Let
  $$
  S'_2 = \left\{j \in [n]: \overline{x}^{\top} A e_j \ge 5/9\right\}.
  $$ The argument to lower bound $|S'_2|$ is similar to the argument
  for $|S'_1|$.  We get
  $$
  \|\tilde{y}_{[n] \setminus S'_2}\| = \Pr_{j \sim \tilde{y}}\left[1 - \overline{x}^{\top} A e_j > 4/9\right] \le \frac{2 \sqrt{s}/3}{4/9} = 3\sqrt{s}/2.
  $$
  and therefore
  $$\|x_{S'_2}\| \ge 1 - t - 3 \sqrt{s}/2 \ge \eps + \alpha,$$
  as desired.  Again, we can truncate $S'_2$ by taking a subset $S_2\subseteq S'_2$ of cardinality $|S'_2|=c_2\log n$. By construction $d(S_1, S_2) \ge 5/9$, and we are done.
\end{proof}

\subsection{Hardness for $\eps$ close to $\frac{1}{2}$}\label{sec:eps-hardness}

To obtain Theorem~\ref{thm:eps-hardness} the main requirement is to
set $\alpha = \frac{1}{2} + O(\eta)$.  The values of $\beta$ and $\gamma$ are
essentially irrelevant in this case -- the only thing needed is
that $\beta$ is bounded away from both $0$ and $\alpha$ and that $\gamma
\le \frac{1}{2}$.

\begin{lemma}
  \label{lem:weight-in-A}
  Let $\alpha = \frac{1}{2} + t$, $\gamma \le \frac{1}{2}$ and $\G$ be the game of
  Reduction~\ref{red:reduction}.  Then for any pair of strategies
  $(x,y)$ with value at least $v_\G(x,y) \ge \alpha - t^2$ it holds
  that $\|x_{[n]}\|$ and $\|y_{[n]}\|$ are both at least $1 - t$.
\end{lemma}

\begin{proof}
  Let $p = \|x_{[n]}\|$ and $q = \|y_{[n]}\|$.  As the value of any
  outcome outside the $\alpha A$ block is at most $\frac{1}{2}$, we have that
  the value of $(x,y)$ is at most
  $$pq \alpha + (1-pq) \frac{1}{2} = t pq + \frac{1}{2},$$
  so that if the value is at least $\alpha - t^2$ we have
  \begin{align*}
  t pq + \frac{1}{2} &\ge \alpha - t^2 \\
  pq &\ge \frac{\alpha - t^2 - 1/2}{t} = 1-t
  \end{align*}
  Since $p, q \in [0,1]$, it follows that they are both at least
  $1-t$.
\end{proof}

\begin{observation}
  \label{obs:value-after-condition}
  Let $(x,y)$ be any pair of strategies with value $v_{\G}(x,y) \ge
  \frac{1}{2}$ and $\|x_{[n]}\| > 0$, $\|y_{[n]}\| > 0$.  Then
  $v_{\G|[n]}(x,y) \ge v_{\G}(x,y)$, provided that $\gamma \le \frac{1}{2}$.
\end{observation}

Plugging this into Lemma~\ref{lem:dense-subgraph}, we can now easily
complete the proof of hardness for $\eps$ close to $\frac{1}{2}$.

\begin{theorem}[Detailed Statement of Theorem~\ref{thm:eps-hardness}]
  \label{thm:eps-hardness-formal}
  For every $\eta>0$ there exist $\delta = \Omega(\eta^2)$, $\alpha
  \ge \frac{1}{2}$ and universal constant $C$ not depending on $\eta$ such
  that the following holds.  Given a graph $G = (V,E)$ we can
  in randomized polynomial time construct a bimatrix game $\G$ with maximum value $\alpha$ (over all strategy pairs) such
  that, if $G=G(n,\frac{1}{2})$ with a hidden clique of size $C\log n$, the following holds (w.h.p.\ over $G$ and $\G$):
  \begin{description}
  \item[Completeness] 
There is a Nash equilibrium
    $(x,y)$ with value $\alpha$.
  \item[Soundness]  
 Given any $(\frac{1}{2}-\eta)$-equilibrium with value $\geq\alpha-\delta$, we can efficiently recover the hidden clique.
  \end{description}
\end{theorem}

\begin{proof}
  Given a graph $G = (V,E)$, we apply Reduction~\ref{red:reduction} with parameters as follows.
  For some $t>0$ to be determined momentarily, let $\alpha = \frac{1}{2} + t$,
  $\beta = \gamma = 1/3$, $\delta = t^2$.

  For the completeness, we shall show that having both players play
  uniformly over the hidden clique is an equilibrium.  For this to
  hold, we have to make sure that there is no row in $B$ with average
  value at least $\alpha$ in the positions corresponding to the
  clique.  By Lemma~\ref{lemma:chernoff} and a union bound over all
  rows of $B$ we can bound the probability of this happening by
  $$N e^{-2 (\alpha-\beta)^2 C \log n} \le N \cdot n^{-C/18}.$$ If $C$
  is a sufficiently large universal constant (e.g., $C = 18 \cdot
  (c_2+1) \log 1/\beta + 1$) this probability is $o(1)$ and the
  completeness property follows.

  For the soundness, consider any $(\frac{1}{2}-\eta)$-approximate equilibrium $(x,
  y)$ with value at least $v_{\G}(x,y) \ge \alpha - \delta=\alpha-t^2$.  By
  Lemma~\ref{lem:weight-in-A} both $\|x_{[n]}\|$ and $\|y_{[n]}\|$ are
  at least $1-t$.  Furthermore, by
  Observation~\ref{obs:value-after-condition} we have $v_{\G|[n]}(x,y)
  \ge \alpha - \delta \ge \alpha(1 - t^2)$.

  Now if $\frac{1}{2}-\eta \le 1 - t - 3t/2 - \alpha = \frac{1}{2} - 7t/4$ we can apply
  Lemma~\ref{lem:dense-subgraph} and extract a dense bipartite
  subgraph of $G$ which can be plugged in to
  Lemma~\ref{lemma:reconstruct} to obtain the hidden clique.  Setting
  $t = 4\eta/7$ we get the result.
\end{proof}

\subsection{Distinguishing Between Low and High Value}\label{sec:delta-hardness}

For Theorem~\ref{thm:delta-hardness} the choices of all three
parameters $\alpha, \beta, \gamma$ of Reduction~\ref{red:reduction}
are important.  We are going to set $\gamma$ close to $0$, and $\alpha
> \beta$ both close to $1$.

On a high level, the proof has the same structure as that of
Theorem~\ref{thm:eps-hardness}.  However, in the current setting
Lemma~\ref{lem:weight-in-A} and
Observation~\ref{obs:value-after-condition} do not apply.  To arrive
at similar conclusions we use a different argument, exploiting the fact that
$(x,y)$ is a $\eps$-equilibrium.  Essentially, the argument is as
follows: the off-diagonal blocks ($B$ and $B^{\top}$) are not stable, since
there is too much incentive for at least one player to
deviate.  Therefore, most of the probability mass in an equilibrium is
concentrated either in the $\alpha A$ block, or in the $\gamma J$
block. However, in the $\gamma J$ block, the value is too small. So,
if the equilibrium has even slightly larger value, its mass must be
concentrated in the $\alpha A$ block.  There it has to actually have
very large value, since otherwise, there is incentive for both players
to deviate to $B$ and $B^{\top}$ to get reward $\beta$.  The rest of the
proof follows as before.

Formally, we have the following lemma, showing that (under certain
conditions) any $\eps$-equilibrium with non-negligible value must
satisfy the conditions of Lemma~\ref{lem:dense-subgraph}.

\begin{lemma}\label{lem:goodvalue-equis}
  Fix a parameter $\eps \in (0,1)$, let $\alpha - \beta \le \eps$, and
  $\gamma = 4 \sqrt{\eps}$ and consider the game $\G$ as in
  Reduction~\ref{red:reduction}.

  Then, w.h.p.\ over $\G$, any $\eps$-equilibrium $(x,y)$ with value
  more than $5 \sqrt{\eps}$ satisfies:
  \begin{itemize}
  \item Both $\|x_{[n]}\|$ and
    $\|y_{[n]}\|$ are at least $1 - \sqrt{\eps}$.
  \item
    $v_{\G|[n]}(x,y) \ge \alpha - 3\eps$.
  \end{itemize}
\end{lemma}

\begin{proof}
  Consider any $\eps$-equilibrium $(x,y)$ with value more than
  $5\sqrt{\eps}$.  Note that this trivially implies that $\eps \le
  1/25$.

  Let $p = \|x_{[n]}\|$ and $q = \|y_{[n]}\|$ be the probability mass
  that the row (resp.\ column) player assigns to the first $n$
  strategies in this equilibrium.  We begin with the first item, i.e.,
  the lower bound on $p$ and $q$.

  Consider the row player's incentive to deviate by shifting the
  probability mass in the first $n$ rows to the uniform distribution
  over the remaining rows.  When the column player is playing outside
  the first $n$ columns, this deviation changes the row player's
  payoff from $0$ to $\gamma$, and when the column player is playing
  in one of the first $n$ columns the row player's payoff decreases by
  at most $\alpha - \beta + o(1)$.  Let us ignore this $o(1)$
  term. Since this is an $\eps$-equilibrium, we have
  \begin{equation}
    \label{eqn:deviation-estimate}
    p(1-q)\gamma-pq(\alpha-\beta)\leq\eps, 
  \end{equation} and so
  \begin{align*}
    p(1-q)&\leq \frac{1}{\gamma}\left(\eps+pq(\alpha-\beta)\right)\\
    &\leq \frac{1}{\gamma}\left(\eps+\alpha-\beta\right)\\
    &\le \sqrt{\eps}/2
  \end{align*}
  Considering also the symmetric argument for the column player, this gives
  \begin{equation}\label{eq:p1-q-bound}
    \max\{p(1-p),q(1-q)\}\leq\max\{p(1-q),q(1-p)\}\leq \sqrt{\eps}/2,
  \end{equation} 
  where the first inequality holds trivially for all $p,q\in[0,1]$.
  The 
   inequality $p(1-p) \le \sqrt{\eps}/2$ together with the
  constraint $\eps \le 1/25$ implies that either $p \le
  \sqrt{\eps}$ or $p \ge 1 - \sqrt{\eps}$.  The inequality for
  $q$ gives an analogous bound.  Moreover, it cannot be the case that
  $p<\frac12<q$ or $q<\frac12<p$, since then we would have
  $\max\{p(1-q),q(1-p)\}\geq\frac14$,
  contradicting~\eqref{eq:p1-q-bound}. Thus, it follows that either
  $p,q\leq \sqrt\eps$ or $p,q\geq 1-\sqrt\eps$.
  
  Let us now exclude the first option.  Suppose for contradiction that
  $p,q\leq \sqrt{\eps}$.  Then the value of the equilibrium is at most
  $$
  pq \alpha + (p(1-q)+q(1-p))\cdot \frac{1}{2} + (1-p)(1-q) \gamma \le
  \eps \alpha + \sqrt{\eps}/2 + \gamma \le \eps \alpha + 4.5 \sqrt{\eps} < 5 \sqrt{\eps},
  $$
  contradicting the assumption that $(x,y)$ has value more than $5
  \sqrt{\eps}$.  Hence both $p$ and $q$ are at least $1 -
  \sqrt{\eps}$.

  It remains to give a lower bound on the conditional value $w :=
  v_{\G|[n]}(x,y)$ obtained when playing inside $\alpha A$.  As above,
  consider the incentive for the row player to deviate to the uniform
  distribution over.  To be more precise, we can use $w$
  instead of $\alpha$ in the bound \eqref{eqn:deviation-estimate}.
  Solving for $w$ gives
  $$
  w \ge \beta + \frac{p(1-q)\gamma - \eps}{pq} \ge \beta - \frac{\eps}{pq} \ge \beta - 2\eps \ge \alpha - 3\eps,
  $$ where we used that $\frac{1}{pq} \le \frac{1}{(1-\sqrt{\eps})^2} \le \frac{25}{16} < 2$ for $\eps \le 1/25 $.
\end{proof}

Equipped with Lemma~\ref{lem:goodvalue-equis}, it is easy to finish the
proof of Theorem~\ref{thm:delta-hardness}.

\begin{theorem}[Detailed statement of Theorem~\ref{thm:delta-hardness}]
  \label{thm:delta-hardness-formal}
  For every constant $\eta>0$ there exist $\eps = \Omega(\eta^2)$ and
  $C = O(1/\eta^3)$ such that the following holds.  Given a
 graph $G$, we can in randomized polynomial time construct a
  bimatrix game $\G$ such that, if $G=G(n,\frac{1}{2})$ with a hidden clique of size $C\log n$, the following holds (w.h.p.\ over $G$
  and $\G$):
  \begin{description}
  \item[Completeness] There is a Nash equilibrium
    $(x,y)$ with both players earning payoff $1-\eta$.
  \item[Soundness] 
    Given any $\eps$-equilibrium with value $\geq\eta$, we can efficiently recover the hidden clique.
  \end{description}
\end{theorem}


\begin{proof}
  Given a graph $G = (V,E)$, we apply Reduction~\ref{red:reduction}
  with parameters as follows.

  Let $\eps=(\eta/5)^2$, $\alpha=1-\eta=1-5\sqrt{\eps}$, $\beta=\alpha-\eps$,
  and $\gamma=4\sqrt{\eps}$.  Assume without loss of generality that
  $\eta$ is small enough so that $\alpha > 3/4$.

  For the completeness, we proceed as in the proof of
  Theorem~\ref{thm:eps-hardness-formal}.  We can upper bound the
  probability that the uniform distribution over the hidden clique is
  not an equilibrium by
  $$ N e^{-2(\alpha-\beta)^2 C \log n} \le N n^{-\Omega(C \eps^2)}.
  $$
  We have $N = n^{O(\log 1/\beta)} = n^{O(\eta)} = n^{O(\sqrt{\eps})}$.
  Letting $C$ be a sufficiently large multiple of $1/\eps^{1.5}$ the
  completeness property follows.

  For the soundness analysis, take any $\eps$-approximate
  equilibrium $(x, y)$ with value at least $5 \sqrt{\eps}$.  By
  Lemma~\ref{lem:goodvalue-equis}, w.h.p.\ $(x,y)$ satisfies the conditions
  of Lemma~\ref{lem:dense-subgraph} with $t = \sqrt{\eps}$ and
  $s = \frac{3\eps}{\alpha} \le 4\eps$.  The only remaining
  thing to check is the condition $\alpha + \eps \le 1 - t -
  3\sqrt{s}/2$, which is easily verified ($1-4\sqrt{\eps}$ is an
  upper bound for the LHS and a lower bound for the RHS).


\end{proof}

\subsection{An Algorithm For Good $\frac{1}{2}$-Approximate Equilibria}\label{sec:1/2-algorithm}

In this section we prove Theorem~\ref{thm:1/2-algo} by describing a
simple algorithm to find a $\frac{1}{2}$-approximate Nash equilibrium with at
least as good value as the best exact Nash equilibrium.  This shows
that the bound on $\eps$ in Theorem~\ref{thm:eps-hardness} is tight.

For general $\frac12$-approximate equilibria (without any constraint
on the value), the following simple algorithm was suggested by
Daskalakis, Mehta and Papadimitiou~\cite{DMP09}. Start by choosing an
arbitrary pure strategy $e_i$ for the row player, let $e_j$ be the column
player's best response to $e_i$, and let $e_k$ be the row player's best
response to $e_j$.  Then the following is a $\frac12$-equilibrium: let
the column player play $e_j$, and let the row player play $e_i$ with
probability $\frac12$ and $e_k$ with probability $\frac12$ (neither
player can gain more than $\frac12$ by deviating, since each player is
playing a best response strategy with probability $\frac12$). Thus,
every bimatrix game has a $\frac12$-approximate equilibrium in which
one of the players plays a pure strategy. We show that this is also
the case for optimal value $\frac12$-equilibria.

\begin{lemma}\label{lem:1/2-structure} For every bimatrix game which has a Nash equilibrium of value $v$, there exists a $\frac12$-approximate equilibrium with value at least $v$ in which one of the players plays a pure strategy.
\end{lemma}

\begin{proof}
Let $M_{\mathrm{row}}$ and $M_{\mathrm{col}}$ be the payoff matrices
for the row and column players, respectively. Let $(x^*,y^*)$ be a
Nash equilibrium in this game of value $v := v_{\G}(x^*, y^*)$, and
let $v_r$ and $v_c$ be the payoff in this equilibrium to the row and
column players, respectively (hence $v=\frac12(v_r+v_c)$).  Without
loss of generality, assume $v_r \ge v_c$ (otherwise, a symmetric
argument applies).

Let the pure strategy $e_j$ be some strategy in the support of $y^*$
for which the \emph{row} player's payoff is at least $v_r$ (when the
row player is playing $x^*$ and the column player $e_j$).  Such a
strategy exists since $v_r$ is the expected payoff for the row player
when the column player plays according to $y^*$.  Furthermore, any
such $e_j$ is a best response to $x^*$ since $(x^*, y^*)$ is an
equilibrium.

Clearly, if the pair of strategies $(x^*, e_j)$ is a $\frac{1}{2}$-equilibrium
we are done, since both the row and colum player are getting at least
the same payoff as for the pair $(x^*, y^*)$.

If $(x^*, e_j)$ is \emph{not} a $\frac{1}{2}$-equilibrium, this must be
because the row player has incentive $\ge \frac{1}{2}$ to deviate (as the
column player is by definition playing a best response).  Note that
this implies that $v_c \le v_r \le \frac{1}{2}$ since the row player's
incentive to deviate can never be more than $1-v_r$.  Let $e_k$ be
some best response for the row player, and consider the pair of
strategies $\left(\frac{1}{2}(x^*+e_k), e_j\right)$.  As above, this
is a $\frac{1}{2}$-equilibrium, since both players are playing a best response
with probability $\frac{1}{2}$.  Furthermore, the payoff for the row player is
at least $v_r/2 + (v_r+1/2)/2 \ge v_r + 1/4 \ge v_r + v_c/2$, and the
payoff for the column player is at least $v_c/2$.  Thus the value is
at least $\frac{1}{2}(v_r+v_c/2+v_c/2) = v$, and we are done.

\end{proof}

Now our algorithm follows immediately.

\begin{proof}[Proof of Theorem~\ref{thm:1/2-algo}]
 Let $M_{\mathrm{row}}$ and $M_{\mathrm{col}}$ be the payoff matrices for the row and column
 matrices, respectively. By Lemma~\ref{lem:1/2-structure} there exists
 such an equilibrium in which one player plays a pure
 strategy. Suppose this is the column player (otherwise, a symmetric
 algorithm applies). Try all possible pure strategies $e_j$ for the
 column player. For each such strategy, solve the following linear
 program (if it is feasible):
\begin{align*}
\mbox{maximize }\quad &\textstyle\frac12x^{\top}(M_c+M_{\mathrm{row}})e_j\\
\mbox{subject to }\quad& x^{\top}M_{\mathrm{col}}e_{j'}\leq x^{\top}M_{\mathrm{col}}e_{j}+\textstyle\frac12&\forall j'\\
& (e_i)^{\top}M_{\mathrm{row}}e_j\leq x^{\top}M_{\mathrm{row}}e_j+\textstyle\frac12&\forall i\\
& x_i\geq0&\forall i\\
& \sum_i x_i=1
\end{align*}

For at least one strategy $e_j$, this LP is feasible and computes a $\frac12$-equilibrium with at least the desired value.
\end{proof}

\comment{Per}{Eden wants to put the proof in the Appendix.}

\comment{Eden}{I do!}

\comment{Per}{But I disagree.}

\section{Finding A Second Equilibrium}\label{sec:uniqueness}

In the following Theorem, $\SD$ refers to the total variation distance
between two vectors, i.e., $\SD(x,y) = \frac{1}{2} \sum |x_i - y_i|$.

\begin{theorem}[Detailed Statement of Theorem~\ref{thm:second_equi}]
  \label{thm:second_equi_formal}
  There is a $C> 0$ such that the following holds for all sufficiently
  small $\eps > 0$.  Given a graph $G$ we can in
  randomized polynomial time construct a bimatrix game $\G'$ which admits a pure Nash equilibrium $(e_i,e_j)$ such that, if $G=G(n,\frac{1}{2})$ with a hidden clique of size $C\log n$,
  the following holds (w.h.p.\ over $G$ and $\G'$):
  \begin{description}

  \item[Completeness] 
There is a Nash equilibrium $(x,y)$ such that
    $\SD(e_i,x) = \SD(e_j,y) = 1$.

  \item[Soundness] 
  Given any
    $\eps$-approximate equilibrium $(x,y)$ of $\G'$ with $\SD(e_i,x)
    \ge \eps + O(\eps^2)$ or $\SD(e_j, y) \ge \eps +
    O(\eps^2)$, we can efficiently recover the hidden clique.
  \end{description}
\end{theorem}

\begin{remark}
  Note that the bound $\eps + O(\eps^2)$ on the statistical
  distance is almost tight: given any true equilibrium $(x,y)$ there are
  $\eps$-approximate equilibria $(x',y')$ with $\SD(x, x') \ge
  \eps$ and $\SD(y, y') \ge \eps$.
\end{remark}

\begin{proof}
  Given $G$, first construct the game $\G$ of
  Theorem~\ref{thm:delta-hardness-formal} with parameter $\eta =
  1/10$, so that given a $(4\eps)$-approximate equilibrium of $\G$
  with value at least $1/10$ we can reconstruct the hidden clique.
  Let $(M_{\mathrm{row}}, M_{\mathrm{col}})$ be the payoff matrices of
  $\G$ and let $N$ denote their dimension.  Consider the new
  $(N+1)$-strategy game $\G' = (M'_{\mathrm{row}},
  M'_{\mathrm{col}})$, with the following payoff matrices.
  \begin{align}
    M'_{\mathrm{row}}=\left(
    \begin{array}{cc}
      M_{\mathrm{row}} & 0\\
      {\bf \lambda} & 1
    \end{array} 
    \right)
    & &
    M'_{\mathrm{col}}=\left(
    \begin{array}{cc}
      M_{\mathrm{col}} & {\bf \lambda}^{\top} \\
      0 & 1
    \end{array} 
    \right),
  \end{align}
  where ${\bf \lambda}$ is the $1 \times N$ vector with each
  coordinate equal to $\lambda$.  We set $\lambda = 8/10$.  $\G'$ is a
  bimatrix game of size $(N+1) \times (N+1)$.

  Clearly, $(e_{N+1}, e_{N+1})$ is a pure Nash equilibrium of $\G'$.
  Note that for any mixed strategy $x$, we can write $\SD(e_{N+1}, x)
  = 1 - x_{N+1}$ so it suffices to obtain good bounds on $x_{N+1}$ and
  $y_{N+1}$.

  Furthermore the completeness case also follows immediately since in
  that case $\G$ has a Nash equilibrium with both players earning payoff $9/10$. 
 As
  $\lambda \le 8/10$ this is an equilibrium in $\G'$ as well and since it
  does not use the $(N+1)$'st strategy we obtain the completeness
  property.


  For the soundness, consider any $\eps$-approximate equilibrium
  $(x,y)$ of $\G'$, let $p = \|x_{[N]}\|$, $q = \|y_{[N]}\|$ be the
  probability that the row (resp.\ column) player plays in the
  original game $\G$ where $x_{N+1} < 1 - \eps - O(\eps^2)$ and
  $y_{N+1} < 1 - \eps - O(\eps^2)$.  We need to show that $(x,y)$ can
  be used to recover the planted clique.

  Let $(\tilde{x}, \tilde{y}) = (x_{[N]}/p, y_{[N]}/q)$ denote the
  strategies conditioned on playing on the first $N$ strategies.  We
  claim that $(\tilde{x}, \tilde{y})$ must be an $\eps'$-approximate
  equilibrium for the original game $\G$, with $\eps' =
  \frac{\eps}{pq}$.  To see this, suppose for contradiction that one
  of the players, say the row player, gains $\eps'$ in $\G$ by
  deviating to some strategy $e_{i'}$.  Consider ``lifting'' this
  to a new strategy $x'$ for $\G'$ (i.e., in the strategy $x'$ the row
  player plays $e_{N+1}$ with probability $1-p$, and 
  $e_{i'}$ with probability $p$).  The change in payoff the row
  player obtains in $\G'$ by switching from $x$ to $x'$ can be written
  as
  \begin{equation}
    \label{eqn:incentive}
    p (q \eps' + (1-q) \cdot 0),
  \end{equation}
  where the $q\eps'$ term is what the row player gains from when the
  column player plays on the first $N$ strategies, and the other term
  is $0$ since, when the column player plays on $N+1$ the row player
  gets the same payoff on all the first $N$ strategies.  As $(x,y)$ is
  an $\eps$-approximate equilibrium in $\G'$, \eqref{eqn:incentive}
  must be bounded by $\eps$ and hence 
  \begin{equation}\label{eq:eps'eps}
  \eps' \le \frac{\eps}{pq}.
  \end{equation}

  \comment{Per}{This last paragraph was a very lengthy way of saying
    something very easy.  Can probably be written better...}

  Now the same argument as in the first part of the proof of
  Lemma~\ref{lem:goodvalue-equis} gives that
  \begin{align*}
    \eps &\ge pq (\lambda - 1) + p(1-q)(1-0) & \text{(row player's incentive to deviate)} \\
    & = p (1 - 11 q / 10) & \text{(by choice of $\lambda$)} \\
    & \ge p (1 - 11 p /10).
  \end{align*} 
  This quadratic inequality in $p$ implies that either $p \le \eps +
  O(\eps^2)$ or $11 p / 10 \ge 1 - \eps - O(\eps^2)$. The first possibility is ruled out by our assumption, therefore $p> \frac{10}{11}(1-\eps-O(\eps^2))$, and similarly $q> \frac{10}{11}(1-\eps-O(\eps^2))$.

  It is easy to see that the value of $(\tilde{x}, \tilde{y})$ must be
  at least $1/10$ because $(x,y)$ is an $\eps$-equilibrium for $\G'$.
  By the choice of parameters, if $(\tilde{x},\tilde{y})$ is also a
  $(4\eps)$-approximate equilibrium for $\G$ then we can reconstruct
  the hidden clique and we are done. But this follows easily from \eqref{eq:eps'eps}.
%
%

\end{proof}

\section{Small support equilibria}\label{sec:small-support}

In this section, we show hardness of finding an $\eps$-approximate Nash equilibrium with small (logarithmic) support when one exists, even for $\eps$ close to $\frac12$. Note that an $\eps$-approximate Nash equilibrium for two-player $n'$-strategy games with support at most $O(\log n'/\eps)$ is guaranteed to exist by the algorithm of Lipton et al.~\cite{LMM03}. Here we consider approximate equilibria with smaller (but still logarithmic) support. Also, note that this is tight, since for $\eps=\frac12$, we have the simple algorithm of~\cite{DMP09}, which gives a $\frac12$ equilibrium of support 3. 

Our reduction for small support equilibria involves the following
construction, which is very similar to the earlier one. 
\boxit{\begin{reduction}
    \label{red:smallsupport}
    Given a graph $G = (V,E)$ with adjacency matrix $A$, construct a game $\G = (M_{\mathrm{row}}, M_{\mathrm{col}})$ as follows.
    \begin{align}\label{eq:small-support-payoffs}
      M_{\mathrm{row}}=\left(
      \begin{array}{cccc}
        \alpha A & 0 & \ldots & 0\\
        B \\
        \vdots & \multicolumn{3}{c}{R}  \\
        B & &
      \end{array} 
      \right)
      & &
      M_{\mathrm{col}}=\left(
      \begin{array}{cccc}
        \alpha A & B^{\top} & \ldots & B^{\top}\\
        0 \\
        \vdots & \multicolumn{3}{c}{J-R} \\
        0
      \end{array} 
      \right),
    \end{align}
    where $B$ is an $N_1 \times n$ matrix whose entries are
    i.i.d.\ $\{0,1\}$ variables with expectation $\beta$.  As in
    Reduction~\ref{red:reduction}, $N_1 = n^c$ where $c = (c_2+1) \log
    1/\beta$ and $c_2$ is the constant from
    Lemma~\ref{lemma:reconstruct}.

    Each payoff matrix of \eqref{eq:small-support-payoffs} contains
    $N_2$ identical copies of $B$, and $R$ is an $N_1N_2 \times
    N_1N_2$ uniformly random $\{0,1\}$ matrix.
  \end{reduction}
}
    
\begin{theorem}   
  For every $\eta>0$ there exists $C>0$ such that finding a
  $(\frac12-\eta)$-equilibrium with support at most $(\log n)/2$ is as
  hard as finding a hidden clique of size $C\log n$ in $G(n,\frac{1}{2})$.
\end{theorem}

\begin{proof}
Given a graph $G$ from $G(n, \frac12)$ (possibly with a hidden
clique), construct the game $\G$ as in
Reduction~\ref{red:smallsupport} with the following parameters: let
$\alpha=\frac{1}{2}+\eta/8$, let $\beta= \alpha + (\frac{1}{2} - \eta) - \eta^2/8 = 1
- \frac{7}{8}\eta - \eta^2/8$.  We choose the dimension $N_2$ as $N_2
= n^{c'}$, where $c'$ is chosen to satisfy $(\eta^2/8)^2 c' = 4c$.
Since $c = (c_2+1) \log 1/\beta = \Theta(\eta)$ we have that $c' =
\Theta(1/\eta^3)$.  We choose the density $C$ of the hidden clique to
be $C = c'/2+1$.

Note that the number of strategies is $n+N_1N_2=n+n^{c' + c}$, so that
we are looking for an equilibrium with support at most
$\frac12\log(n+n^{c' + c})=C'\log n$ for some $c'/2 \le C' \le c'/2 +
1$ (assuming $\eta$ is sufficiently small).

The completeness follows easily. Suppose $G$ contains a hidden clique
of size $C\log n>C'\log n$. Then if both players play uniformly over
the same subset of $C'\log n$ clique vertices, they both
achieve reward $\alpha$. The probability that, say, the row player can
gain $\frac{1}{2}-\eta$ (i.e.\ get payoff
$\alpha+\frac{1}{2}-\eta=\beta+\eta^2/8$) by deviating to some
row in $B$ (note that he can only deviate to rows in copies of $B$) is by the Chernoff bound Lemma~\ref{lemma:chernoff} and a union bound at most
$$N_1e^{-(\eta^2/8)^2 C '\log
  n} = n^{c - (\eta^2/8)^2 C'} \le n^{c - (\eta^2/8)^2 c'/2} = n^{-c}.$$

Now, for the soundness, consider any $(\frac{1}{2}-\eta)$-equilibrium
$(x,y)$. Let us first show that both players must have most of their
probability concentrated in the $\alpha A$ block. Let $p =
\|x_{[n]}\|$ and $q = \|y_{[n]}\|$ (the probabilities that each player
plays in the first $n$ rows/columns). Let us consider the two player's
incentive to deviate. In the $\alpha A$ block, the row player achieves
at most payoff $\alpha$, and can achieve payoff $\beta-o(1)$ by
playing uniformly over all rows in $B$ (w.h.p.\ this is true for all
distributions over columns in $[n]$). In particular, there exists at
least one row in $B$ in which the row player can achieve this
value. Now consider the right hand side of the payoff matrix. Let
$\lambda\in[0,1]$ be the payoff that the row player receives in $R$
(thus, the column player receives $1-\lambda$ here). For any row in
$B$, there are $N_2$ corresponding rows in $R$ (one corresponding to
each copy of $B$).  Since the column player's support is at most
$C'\log n$, the probability that regardless of the column player's
choice of support, there will be at least one row among these that has
all 1's in the corresponding positions is at least
\begin{align*}
  1-n^{C'\log n}(1-2^{-C'\log n})^{N_2} &\geq 1- n^{C'\log n}\exp(-2^{-C' \log n} N_2) \\
  & = 1- n^{C'\log n}e^{-n^{-C'+c'}}\\
  &\geq 1 - n^{C'\log n}e^{-n^{c'/2-1}}.
\end{align*}
Thus, w.h.p.\ for every row in $B$ and every possible strategy for the column player (up to the restriction on support size), there is a row in $R$ corresponding to the correct row in $B$ s.t.\ the row player would achieve payoff 1 in $R$ (by deviating to this row). In particular, this is true for the row in $B$ where the row player can achieve value $\beta - o(1)$ (as before, we will ignore this $o(1)$). To summarize, by deviating, the row player can gain at least $$qp(\beta-\alpha)+(1-q)(1-(1-p)\lambda).$$ Similarly, the column player's incentive to deviate is at least $$pq(\beta-\alpha)+(1-p)(1-(1-q)(1-\lambda)).$$ On average, the two players' incentive to deviate is at least
$$
pq(\beta-\alpha)+\textstyle\frac12((1-p)+(1-q)) - \textstyle\frac12(1-p)(1-q) = pq(\beta-\alpha) + \textstyle\frac12(1-pq),
$$
and since this incentive is at most $\frac12-\eta$, we have
\begin{equation}\label{eq:small-support-prob}
pq(\beta-\alpha-\textstyle\frac12)+\textstyle\frac12 \leq \textstyle\frac12 - \eta,
\end{equation}
or
\begin{align*}
pq\geq \eta/(\textstyle\frac12-(\beta-\alpha)) = \eta/(\eta+\eta^2/8)> 1 - \eta/8.
\end{align*}

Now it remains to bound the conditional value $w$ that is achieved in the $\alpha A$ block. This we can do using the same analysis as above, but substituting $w$ for $\alpha$. Making this substitution in \eqref{eq:small-support-prob} and solving for $w$ we have 
\begin{align*}
w\geq \eta/pq + \beta-\textstyle\frac12&=\eta/pq + \textstyle\frac12-7\eta/8-\eta^2/8\\
&\geq \textstyle\frac12 +\eta/8 - \eta^2/8\\
&= \alpha(1- \eta^2/(4+\eta)) \geq \alpha(1-\eta^2/4).
\end{align*}

We are now in a position to apply Lemma~\ref{lem:dense-subgraph}.  As
stated Lemma~\ref{lem:dense-subgraph} only applies to
Reduction~\ref{red:reduction} and not the present reduction but it can
be verified that it works also in this case (what is needed is that
the size of $B$ is the same, and the presence of an approximate
equilibrium with most mass in $A$ and good value in $A$).  We have $t
= \eta/8$ and $s = \eta^2/4$ which is easily checked to satisfy the
condition $1 - t - 3\sqrt{s}/2 \ge \alpha + (\frac{1}{2}-\eta)$.  Thus we
conclude the existence of a dense bipartite subgraph which, as before,
allows us to reconstruct the hidden clique.

\end{proof}

Note that we have a much smaller gap between the completeness and hardness above than in the other problems we have considered. In particular, we do not claim that finding a $\frac12-\eta$-equilibrium with small support is hard even when an exact equilibrium with small support exists. However, modifying parameters in the above proof, such hardness can be shown for a smaller additive approximation:

\begin{theorem}\label{thm:small-support-1/4}   
  For every $\eta>0$ there exists $C>0$ such that finding a
  $(\frac14-\eta)$-equilibrium with support at most $O(\log n)$ 
  in a two-player game which admits a pure Nash equilibrium is as hard as finding a hidden clique of size $C\log n$ in $G(n,\frac{1}{2})$.
\end{theorem}

\comment{Eden}{Is it support $(\log n)/2$ here as well? What is the significance of this factor?}

\newcommand{\Ta}{\Theta}
\newcommand{\De}{\Delta}
\newcommand{\ta}{\theta}
\newcommand{\ra}{\rightarrow}
\newcommand{\cP}{{\mathcal P}}
\newcommand{\cC}{{\mathcal C}}

\section{Computing approximate pure Bayes-Nash equilibrium}
\label{sec:bayes}

Bayesian games model the situation where the players' knowledge of the
world is incomplete.  In this paper we focus on Bayesian games
with two players, but the results generalize to an arbitrary number of
players. More details on Bayesian games can be found in most Game
Theory textbooks, for example in \cite{FudenbergTirole}.

In a \emph{Bayesian game} the payoff of the players depends on the
state of the world in addition to the players' strategies. In a
situation with two players, the row player and the column player, each player is presented with
a signal, called type, about the state of the world $\ta_{\mathrm{row}}\in \Ta_{\mathrm{row}}$
and $\ta_\mathrm{col}\in\Ta_\mathrm{col}$, respectively.  The types are distributed
according to some joint distribution $\cP$ and are not necessarily
independent. The types determine the payoff matrices $M_{\mathrm{row}}(\ta_{\mathrm{row}},\ta_\mathrm{col})$ and $M_\mathrm{col}(\ta_{\mathrm{row}},\ta_\mathrm{col})$. Denote 
the set of rows and columns in this matrix by $S_{\mathrm{row}}$ and $S_\mathrm{col}$, respectively.
 Each player chooses an action $s_{\mathrm{row}}\in S_{\mathrm{row}}$ and $s_\mathrm{col}\in
S_\mathrm{col}$ from their respective set of actions.
The payoff function of the first player is thus $u_{\mathrm{row}}(s_{\mathrm{row}},s_\mathrm{col},\ta_{\mathrm{row}},\ta_\mathrm{col})=M_{\mathrm{row}}(\ta_{\mathrm{row}},\ta_\mathrm{col})_{s_{\mathrm{row}},s_\mathrm{col}}\in
[0,1]$. The payoff function $u_\mathrm{col}$ is defined similarly.
The payoff matrices, that depend on the players' types, as well as the distribution on types is known to the players
ahead of the game. 

A \emph{pure strategy} for the row player  in a Bayesian game is a function (that by a slight abuse of notation) we denote 
by $s_{\mathrm{row}}:\Ta_{\mathrm{row}}\ra S_{\mathrm{row}}$ that for each type $\ta_{\mathrm{row}}$ as observed by row player associates a strategy $s_{\mathrm{row}}(\ta_{\mathrm{row}})$
that the player chooses to execute. A pure strategy $s_\mathrm{col}:\Ta_\mathrm{col}\ra S_{\mathrm{row}}$ is defined similarly. 

Denote by $\cP_{\ta_{\mathrm{row}}}$ the distribution on player column player's types $\ta_\mathrm{col}$ conditioned on the type $\ta_{\mathrm{row}}$ being observed. 
For a pair of pure strategies $(s_{\mathrm{row}},s_\mathrm{col})$ the payoff function of the row player  is given by 
$$
p_{\mathrm{row}}(\ta_{\mathrm{row}}) = \E_{\ta_\mathrm{col}\sim \cP_{\ta_{\mathrm{row}}}} [ u_{\mathrm{row}} ( s_{\mathrm{row}}(\ta_{\mathrm{row}}), s_\mathrm{col}(\ta_\mathrm{col}), \ta_{\mathrm{row}}, \ta_\mathrm{col})]. 
$$

A \emph{pure strategy Nash equilibrium} in a Bayesian game, is a pair of functions $s_{\mathrm{row}}$, $s_\mathrm{col}$ such that for all types
observed, neither player has an incentive to deviate from his current strategy. In other words, for each $\ta_{\mathrm{row}}$, and 
for each $s_{\mathrm{row}}'\in S_{\mathrm{row}}$, 
$$
p_{\mathrm{row}}(\ta_{\mathrm{row}}) \ge \E_{\ta_\mathrm{col}\sim \cP_{\ta_{\mathrm{row}}}} [u_{\mathrm{row}} ( s_{\mathrm{row}}', s_\mathrm{col}(\ta_\mathrm{col}), \ta_{\mathrm{row}}, \ta_\mathrm{col})],
$$
and a similar condition holds for $p_\mathrm{col}$. 

Since a pure Nash equilibrium need not exist in non-Bayesian games, it
need not exist in Bayesian games either. Moreover, while verifying
whether a non-Bayesian two player game has a pure Nash equilibrium is
trivial, verifying whether a pure Bayesian Nash equilibrium exists is
NP-hard \cite{CS08}. Furthermore, as the example in \cite{CS08}
demonstrates, this problem remains hard even when the distribution on
types is uniform and the payoff does not depend on the players' types.

A \emph{pure  $\eps$-Bayesian Nash equilibrium ($\eps$-BNE)} is defined similarly to an $\eps$-Nash equilibrium. 
For each observed type $\ta_{\mathrm{row}}$, the incentive to deviate should be bounded by $\eps$:
$$
p_{\mathrm{row}}(\ta_{\mathrm{row}}) > \E_{\ta_\mathrm{col}\sim \cP_{\ta_{\mathrm{row}}}} [u_{\mathrm{row}} ( s_{\mathrm{row}}', s_\mathrm{col}(\ta_\mathrm{col}), \ta_{\mathrm{row}}, \ta_\mathrm{col})]-\eps.
$$
A similar requirement should hold for the column player. 

We show that for general distributions on types and for some small constant $\eps$, finding a pure $\eps$-BNE 
in games where a pure BNE exists is still NP-hard. On the other hand, we also show that if the distribution on the
players' types is uniform, whenever a pure BNE exists,  a pure $\eps$-BNE can be found in quasi-polynomial time. 

\subsection{General distributions on types}

We show that for some constant $\eps$, determining whether a pure strategy $\eps$-Bayes Nash equilibrium exists is NP-hard. 
Specifically, $\eps = 0.004$ suffices. We prove:

\begin{theorem}
\label{thm:BNEhard}
Let $\eps=0.004$. Then given a Bayesian game that admits a pure BNE,
it is NP-hard to find a pure $\eps$-BNE for the game.  Moreover, it is
NP-hard to solve the promise problem of distinguishing games that
admit a pure BNE from games that do not admit a pure $\eps$-BNE.
\end{theorem}


\begin{proof}
We give a reduction from the problem of $3$-coloring $4$-regular
graphs, which is known to be NP-complete \cite{dailey1980uniqueness}.
Let $G=(V,E)$ be a $4$-regular graph with $|V|=n$. The edges of $G$
can be properly colored with $5$ colors using Vizing's algorithm. In
other words, we can compute a coloring $c:E\ra\{1,\ldots,5\}$ such
that for every two edges $e_1,e_2$ incident to the same vertex,
$c(e_1)\neq c(e_2)$.

We now design a Bayesian game such that:
\begin{itemize}
\item 
if $G$ is $3$-colorable, then the game admits a pure BNE;
\item 
if $G$ is not $3$-colorable, then the game admits no pure $\eps$-BNE.
\end{itemize}
The type of each of the players corresponds to a vertex in the graph, thus $\Ta_{\mathrm{row}}=\Ta_\mathrm{col}=V$. 
The distribution on types is such that with probability $4/5$, $(\ta_1,\ta_2)$ is a random edge in $E$;
with probability $1/5$, $(\ta_1,\ta_2)=(\ta,\ta)$ is the same random vertex in $V$. We note that 
in this example there is a high degree of correlation between the types of the two players.

Each player has $6$ strategies available to him: $S_{\mathrm{row}}=S_\mathrm{col}=\{1,2,3\}\times\{0,1\}$. The first 
coordinate of $s_{\mathrm{row}}(v)$ should be thought of as the ``color'' assigned by player $A$ to vertex $v$, 
while the second coordinate is either ``head'' or ``tails'', the use of which will be explained later. 
We denote the first coordinate by $s_{\mathrm{row}}(v)_1 \in \{1,2,3\}$ and the second by $s_{\mathrm{row}}(v)_2\in \{0,1\}$. 

The payoff functions are defined as follows ($s_{\mathrm{row}}$ and $s_\mathrm{col}$ stand for $s_{\mathrm{row}}(\ta_{\mathrm{row}})$ and $s_\mathrm{col}(\ta_\mathrm{col})$, respectively):

\medskip
\begin{tabular}{|c|c|c|c|c|}
\hline
$\ta_{\mathrm{row}} \stackrel{?}{=} \ta_\mathrm{col}$ & $(s_{\mathrm{row}})_1\stackrel{?}{=} (s_\mathrm{col})_1$ & $(s_{\mathrm{row}})_2 \stackrel{?}{=} (s_\mathrm{col})_2$ & $u_{\mathrm{row}}$ & $u_\mathrm{col}$ \\ \hline
yes & yes & & $1$ & $1$\\ \hline
yes & no & yes & $0.64$ & $0$ \\ \hline
yes & no & no & $0$ & $0.64$ \\ \hline
no & no & & $1$ & $1$\\ \hline
no & yes & yes & $0.01\times 2^{c(\ta_{\mathrm{row}},\ta_\mathrm{col})}$ & $0$ \\ \hline
no & yes & no & $0$ & $0.01\times 2^{c(\ta_{\mathrm{row}},\ta_\mathrm{col})}$ \\ \hline
\end{tabular}
\medskip

In other words, the first coordinate of the strategy represents vertex color. The players
are rewarded for producing a consistent $3$-coloring: for the same vertex $\ta$ the colors
$(s_{\mathrm{row}}(\ta))_1$ and $(s_\mathrm{col}(\ta))_1$ should match; for different vertexes connected by an edge
the colors $(s_{\mathrm{row}}(\ta_{\mathrm{row}}))_1$ and $(s_\mathrm{col}(\ta_\mathrm{col}))_1$ should differ. If these conditions are satisfied, 
both players are rewarded with a payoff of $1$. If a certain edge or vertex fails to satisfy the conditions,
the games becomes a zero-sum games that depends on the second coordinate of the players' strategies. 
The payoff can take values of $0.02$, $0.04$, $0.08$, $0.16$, $0.32$ for edge pairs $(\ta_{\mathrm{row}},\ta_\mathrm{col})$ and the 
value of $0.64$ for vertex pairs $(\ta,\ta)$. What makes these values interesting is that no combination 
of these values with coefficients of $-1,0,+1$ adds up to less than $0.02$. 

If the graph $G$ admits a $3$-coloring $\cC:V\ra \{1,2,3\}$, then the pure strategy $s_{\mathrm{row}}(v)=s_\mathrm{col}(v)=(\cC(v),0)$ yields 
the optimal possible payoff of $1$ for both players, and thus is a pure BNE. 

On the other hand, suppose that $G$ is not $3$-colorable. Let $s_{\mathrm{row}}(v)$, $s_\mathrm{col}(v)$ be a set of pure strategies that is an $\eps$-BNE.
Denote the expected payoff functions  under these strategies by $u_{\mathrm{row}}(v)$ and $u_\mathrm{col}(v)$. If 
the row player deviates in her strategy from $(s_{\mathrm{row}}(v)_1,s_{\mathrm{row}}(v)_2)$ to $(s_{\mathrm{row}}(v)_1,1-s_{\mathrm{row}}(v)_2)$ her payoff only changes
for type pairs where the conditions are not satisfied. Let us denote this change by $\De_{\mathrm{row}}(v)$. The contribution
of each edge $(v_{\mathrm{row}},v_\mathrm{col})$ to $\De_{\mathrm{row}}(v_{\mathrm{row}})$ and $\De_\mathrm{col}(v_\mathrm{col})$ cancels out, and hence, summing over all edges,
\begin{equation}
\label{eq:BNE}
\sum_{v\in V} \left(\De_{\mathrm{row}}(v)+\De_\mathrm{col}(v)\right) = 0.
\end{equation}
By the way we designed our payoffs, the values of $\De_{\mathrm{row}}(v)$ and $\De_\mathrm{col}(v)$ are integer multiples of $0.004$. 
Moreover, if at least one edge adjacent to $v$ is not satisfied, then $\De_{\mathrm{row}}(v)$ cannot be equal to $0$. Since
$G$ is not $3$-colorable, it means that $\De_{\mathrm{row}}(v)\neq 0$ for at least one of the vertexes. By \eqref{eq:BNE}
this means that $\De_{\mathrm{row}}(v)>0$ or $\De_\mathrm{col}(v)>0$ and thus $\De_{\mathrm{row}}(v)\ge 0.004$ or $\De_\mathrm{col}(v)\ge 0.004$. This contradicts 
the assumption that $(s_{\mathrm{row}},s_\mathrm{col})$ is an $\eps$-BNE.
\end{proof}

\subsection{Uniform distribution on types}

In this section we show that in the case where the distribution on types is uniform, a pure $\eps$-BNE can be computed 
in quasi-polynomial time. This contrasts with the previously noted fact from \cite{CS08} that computing 
a pure BNE is NP-hard even in this special case. As for other quasi-polynimial time computable approximate equilibria we've considered, whose exact variants are NP-hard, this problem is also as hard as Hidden Clique:

\begin{theorem}  For every $\eta>0$, finding a $(\frac14-\eta)$-approximate pure BNE in a two-player Bayesian games with uniformly
distributed types and in which a pure BNE exists  
 is as hard as finding a hidden clique of size $C\log n$.
\end{theorem}
\begin{proof} This follows immediately from Theorem~\ref{thm:small-support-1/4}. Let $M_{\mathrm{row}}$ and $M_{\mathrm{col}}$ be as in Reduction~\ref{red:smallsupport} with parameters set in order to achieve the $\frac14-\eta$-hardness of Theorem~\ref{thm:small-support-1/4}. Consider the Bayesian game in which both players have exactly $C'\log n$ types,
\\
\comment{Eden}{Again, is this $\log n/2$?}
\\
 the distribution over types is uniform, and the payoff matrices for each player are always \emph{the same} $M_{\mathrm{row}}$ and $M_{\mathrm{col}}$ as above. Note that when there is a hidden clique, the following is an pure BNE: choose a subset $S$ of $C'\log n$ clique vertices, and let each player play according to a one-to-one mapping from their type set to $S$. On the other hand, it is easy to see that every pair of pure strategies for the Bayesian game corresponds to mixed strategies for the original game, where the incentive to deviate for, say, the row player, is the expected incentive to deviate in the Bayesian game (over all choices of types for the row player). In particular, if a player always has incentive at most $\frac14-\eta$ to deviate in the Bayesian game, their incentive will be at most $\frac14-\eta$ in the original game, and so by Theorem~\ref{thm:small-support-1/4}, we can recover the hidden clique.
\end{proof}

We also note that our result generalizes to the case when the distribution 
on types is a product distribution, i.e. when $\ta_{\mathrm{row}}$ is independent from $\ta_\mathrm{col}$. To simplify the presentation we  
further assume that the type space is of equal size for both players, i.e. $|\Ta_{\mathrm{row}}|=|\Ta_\mathrm{col}|=k$. 
We prove:

\begin{theorem}
\label{thm:BNEeasy} In a two-player Bayesian game,
suppose that the types are distributed uniformly on the space $\Ta_{\mathrm{row}}\times\Ta_\mathrm{col}$, and that $|\Ta_{\mathrm{row}}|=|\Ta_\mathrm{col}|=k$, and $|S_{\mathrm{row}}|=|S_\mathrm{col}|=n$. 
Assuming that a pure BNE exists, we can find a pure $\eps$-BNE in time $n^{O((\log n+\log k)/\eps^2)}$. 
\end{theorem}

\begin{remark}
  The assumption in Theorem~\ref{thm:BNEeasy} can be relaxed to a pure $(\eps/2)$-BNE equilibrium existing (instead of an 
  actual equilibrium).
\end{remark}

\begin{proof}
The proof is similar in spirit to the quasi-polynomial $\eps$-Nash algorithm of \cite{LMM03}, but some additional work 
is needed. We first (approximately) guess the payoffs and the allowed strategies for both players for all possible types. We then use linear programming to produce a (not necessarily pure) $(3\eps/4)$-BNE. We then sample from this approximate non-pure BNE to obtain a
pure $\eps$-BNE. 

For simplicity, denote $\Ta_{\mathrm{row}}=\Ta_\mathrm{col}=\{1,\ldots,k\}$ and $S_{\mathrm{row}}=S_\mathrm{col} = \{1,\ldots,n\}$. Let $s_{\mathrm{row}}:\Ta_{\mathrm{row}}\ra S_{\mathrm{row}}$, $s_\mathrm{col}:\Ta_\mathrm{col}\ra S_\mathrm{col}$ be 
a pair of pure equilibrium strategies (that we assume exist). Let $p^{\mathrm{row}}_{ij}$ and $p^\mathrm{col}_{ij}$ be the corresponding 
 payoff values. In other words, $p^{\mathrm{row}}_{ij}$ is the payoff the row player gets when his type is $i\in\{1,\ldots,k\}$ and he plays strategy $j\in\{1,\ldots,n\}$. The equilibrium assumption is that for all $i$,
 $$
 p^{\mathrm{row}}_{i,s_{\mathrm{row}}(i)} = \max_j p^{\mathrm{row}}_{ij},
 $$
 and a similar condition holds for $p^\mathrm{col}$. We first claim that we can recover all the values of $p^{\mathrm{row}}$ and $p^\mathrm{col}$ within 
 an error of $\eps/4$ with probability $>n^{-O((\log (nk))/\eps^2)}$.
\begin{claim}
\label{cl:BNE1}
There is a polynomial time algorithm that with probability $>n^{-O((\log (nk))/\eps^2)}$ outputs values $q^{\mathrm{row}}_{ij}$, $q^\mathrm{col}_{ij}$ such that for all $i,j$, 
$|q^{\mathrm{row}}_{ij}-p^{\mathrm{row}}_{ij}|<\eps/8$ and $|q^\mathrm{col}_{ij}-p^\mathrm{col}_{ij}|<\eps/8$.
\end{claim}

\begin{proof}
We show how to approximate $p^{\mathrm{row}}$ well with probability $>n^{-O((\log (nk))/\eps^2)}$. The claim follows by approximating $p^{\mathrm{row}}$ and $p^\mathrm{col}$ independently.
Pick a subset of $m=\lceil (40\log (nk))/\eps^2\rceil$ of the column player's types $t_1,\ldots,t_m\in\{1,\ldots,k\}$. For each type $t_r$ guess the column player's strategy $s_r$ on this type. With probability $n^{-m}$ we correctly guess all strategies, i.e. $s_r = s_\mathrm{col}(t_r)$ for all $r$. Set 
$$
q^{\mathrm{row}}_{ij} = \frac{1}{m}\sum_{r=1}^{m} u_{\mathrm{row}}(j,s_r,i,t_r).
$$
In other words, we calculate an estimate of the expected payoff $p^{\mathrm{row}}$ using only the value at the types we've guessed. By 
Hoeffding's Inequality, for each $(i,j)$ the probability 
$$
{\mathbf P}[|q^{\mathrm{row}}_{ij}-p^{\mathrm{row}}_{ij}|\ge \eps/8] \le 2\cdot \exp(-m\eps^2/32) < 1/(4nk).
$$
The claim follows by union bound. 
\end{proof}

Allowing for a blow-up of $n^{O((\log (nk))/\eps^2)}$ in running time, we may assume from now on that the correct values 
of $q^{\mathrm{row}}$ and $q^\mathrm{col}$ had been computed (up to an error of $\eps/8$), since in the end we can check whether the set of pure strategies obtained 
is a $\eps$-BNE.

Next, we formulate a linear program that obtains a $(3\eps/4)$-BNE with payoff values close to $q^{\mathrm{row}},q^\mathrm{col}$. The variables 
of the program are $X^{\mathrm{row}}_{ij}$ and $X^\mathrm{col}_{ij}$, where $X^{\mathrm{row}}_{ij}$ corresponds to the probability that the row player plays strategy 
$j$ on type $i$. For each type $i$ let $$M^{\mathrm{row}}_i := \max_j q^{\mathrm{row}}_{ij}$$ be the maximum possible payoff the row player can attain 
on type $i$ under payoffs $q^{\mathrm{row}}_{ij}$. We only allow $X^{\mathrm{row}}_{ij}$ to be non-zero when the corresponding payoff 
$q^{\mathrm{row}}_{ij}> M^{\mathrm{row}}_i - \eps/8$. This guarantees that the solution we obtain is a $(3\eps/4)$-BNE. In addition, we enforce
the payoffs to be close to $q^{\mathrm{row}}_{ij}$. For each $i,j$ we have the constraint
$$
\left| q^{\mathrm{row}}_{ij} - \frac{1}{k}\cdot \sum_{y,z} X^\mathrm{col}_{yz} \cdot u_{\mathrm{row}}(j,z,i,y) \right| < \eps/4,
$$ 
and a similar condition on the $X^{\mathrm{row}}$'s. 
This linear program is feasible since the pure BNE we assumed exists is a solution to it. Denote the 
resulting expected payoffs by $v^{\mathrm{row}}_{ij}$ and $v^\mathrm{col}_{ij}$. 

The solution thus obtained is a $(3\eps/4)$-BNE. To see this, we actually observe that  the equilibrium is well supported: for each type $i$, 
each strategy $j$ with $X^{\mathrm{row}}_{ij}\neq 0$ the payoff $v^{\mathrm{row}}_{ij}\ge M_i^{\mathrm{row}}-3\eps/8 \ge \max_{j'} v^{\mathrm{row}}_{ij'}-5\eps/8$. 

To complete the proof we now obtain a set of pure strategies $s_{\mathrm{row}}':\Ta_{\mathrm{row}}\ra S_{\mathrm{row}}$, $s_\mathrm{col}':\Ta_\mathrm{col}\ra S_\mathrm{col}$ by sampling 
$j=s_{\mathrm{row}}'(i)$ according to the probability distribution $X^{\mathrm{row}}_{ij}$. Assuming $k=\Omega((\log (nk))/\eps^2)$, once
again by Hoeffding's Inequality the payoffs will be  $\eps/8$-close to the payoffs $v^{\mathrm{row}}_{ij}$, and thus 
the resulting game will be in an $\eps$-BNE. 

In the case when  $k=O((\log (nk))/\eps^2)$, i.e. the number of types is small, we can find the exact BNE by brute force, completing 
the proof of the theorem. 
\end{proof}

\bibliographystyle{alpha}
\bibliography{nashvalue}

\end{document}